\documentclass{amsart}
\usepackage{amsmath, amsthm, amscd, amsfonts, amssymb, color}
 \newtheorem{thm}{Theorem}[section]
 \newtheorem{cor}[thm]{Corollary}
 \newtheorem{lem}[thm]{Lemma}
 
 \theoremstyle{definition}
 \newtheorem{defn}[thm]{Definition}
 
  \newtheorem{op}[thm]{Open Problem}
 \theoremstyle{remark}

 \numberwithin{equation}{section}

\begin{document}



\title[On the presence of families of pseudo-bosons in nilpotent Lie...]{On the presence of families of pseudo-bosons in nilpotent Lie algebras of arbitrary corank }

\author[F. Bagarello]{Fabio Bagarello}
\address{Scuola Politecnica dell'Universit\'a \endgraf
Dipartimento di Energia, Ingegneria dell'Informazione e modelli Matematici \endgraf
Universit\'a degli Studi di Palermo\endgraf
Viale delle Scienze, I-90128\endgraf
Palermo, Italy\endgraf}
\email{fabio.bagarello@unipa.it}

\author[F.G. Russo]{Francesco G. Russo}
\address{Department of Mathematics and Applied Mathematics\endgraf
 University of Cape Town \endgraf
 Private Bag X1, Rondebosch 7701\endgraf 
 Cape Town, South Africa\endgraf}
\email{francescog.russo@yahoo.com}


\date{\today}


\begin{abstract}
\noindent We have recently  shown that pseudo-bosonic operators realize concrete examples of finite dimensional nilpotent Lie algebras over the complex field. It has been the first time that such operators were analyzed in terms of nilpotent Lie algebras (under prescribed conditions of physical character). On the other hand, the general classification of a finite dimensional nilpotent Lie algebra $\mathfrak{l}$ may be given via the size of its Schur multiplier involving  the so-called corank $t(\mathfrak{l})$ of $\mathfrak{l}$. We represent $\mathfrak{l}$ by pseudo-bosonic ladder operators for $t(\mathfrak{l}) \le 6$ and this allows us to represent  $\mathfrak{l}$ when its dimension is $\le 5$. 
\end{abstract}

\date{\today}

\maketitle

\section{Introduction}


The notion of Heisenberg algebra and of Heisenberg group appear since long time in the description of  quantum dynamical systems and sophisticated techniques of algebra, topology and functional analysis are often used, in order to understand the behaviour of elementary particles or, more in general, of some quantum systems. When it is possible to detect algebraic structures like nilpotent Lie algebras, we are in a very good position to recognize symmetries and conservation laws of significant interest in physics. This is exactly the content of the well known Noether's Theorem, \cite{gold}. The reference of Snobl and others  \cite{snobl} show some recent techniques in the context of Lie algebras and theoretical physics, but here we will focus on a specific classification, due to the so called ``Berkovich--Zhou invariant'' (or ``corank'') of nilpotent finite dimensional Lie algebras (over a field of characteristic $0$).

We will refer to the recent works \cite{nr1, nr2, nr3}, which generalize \cite{bat, es, es2}, but, at the same time, the present contribution continues the efforts we made in \cite{bagrus2018}, where we looked at dynamical systems involving pseudo-bosons in terms of finite dimensional nilpotent Lie algebras for values of the dimension smaller than five.

 During our discussions, we will list a first classification 
 of finite dimensional nilpotent Lie algebras by the size of
 their dimension, and, specifically, when the dimension is at
 most $5$ (see \cite{degraaf, morozov, turk} and Theorem
 \ref{classification} below). Then we will use  the notion of
 \textit{Schur multiplier} of a  Lie algebra, referring to \cite{weibel} for its classical approach via homological algebra.  Schur \cite{schur} introduced this concept in 1907 in group theory, but his approach was subject to generalizations of various nature for different algebraic and geometric structures.

 The Schur multiplier does not seem to be related with the classification of finite dimensional nilpotent Lie algebras by dimension, but this is only apparently true. Indeed we
 may define a non-negative integer (which turns out to be a
 numerical invariant), called ``corank''  and due to
 Berkovich and Zhou \cite{berkovich, zhou}. This number
 measures the size of the Schur multiplier of a nilpotent Lie
 algebra and we will see that the corank allows us to get a
 second classification of nilpotent Lie algebras, but this
 time the classification gives information on the way in
 which our Lie algebra may be written as the sum of two
 smaller abelian Lie algebras.

 This second classification is more useful if we
 want to focus on the structure of a general Lie algebra.
 Moreover,  this new classification presents again nilpotent
 Lie algebras of dimension $\le 5$ when the value of the
 corank is smaller than $6$ (see Theorem \ref{tl} below
 for more details).

 While the parallel approach of classification via corank and
 via dimension was known from a pure algebraic--geometric perspective (see
 \cite{bat, es, es2, mon, nr1, nr2}), one of the
 interesting result of the present paper is that explicit
 physical examples of dynamical systems can be related to
 this classification. This point is not known in literature,
 to the best of our knowledge.
 
  In fact we will show that all finite dimensional nilpotent
 Lie algebras, over the field of complex numbers may be
 realized  in terms of (pseudo-) bosonic operators  when
 their corank is at most $6$ (see Theorem \ref{main2}
 below). Moreover we will provide examples of significant
 interest in physics when the value of the corank is
 arbitrary big (see Theorem \ref{main1}).  
 

Section 2 recalls some well known results and definitions of algebraic nature. They can be found in \cite{beltita,snobl,weibel} and are integrated by some recent material in \cite{nr1, nr2, nr3}. 
 We refer to \cite[Section 3]{bagrus2018} for the terminology and notations which will be involved in the proofs of our main results, where a framework of functional analysis and of theory of the operators is applied. Finally, Section 3 illustrates  the main results with corresponding proofs.

\section{Some results of low dimensional cohomology}

In order to focus on the main topics, without dispersion of energies in additional details, we do not report the classical notion of \textit{Chevalley-Eilenberg complex} for a Lie algebra, since it can be found in \cite[Definition 7.7.1]{weibel}. Analogously, we will assume that the reader is familiar with the second cohomology Lie algebra with complex coefficients. This Lie algebra is widely described in \cite[Corollary 7.7.3 and Exercise 7.7.5]{weibel} and  we will call it  \textit{ Schur multiplier} of the Lie algebra $\mathfrak{l}$, using the notation $M(\mathfrak{l})$. In fact we will work only with finite dimensional  Lie algebras over the field of complex numbers.

The size of the Schur multiplier describes how many extensions we can actually form, beginning from an abelian Lie algebra with another one acting on it. Furthemore, we will see  that it is possible to introduce a numerical invariant, which we will call  ``corank'', or  ``Berkovich--Zhou invariant'', and this measures the size of the Schur multiplier from some ideal cases for which we known essentially all. The main results of this paper will indeed show that it is possible to find concrete examples for any possible value of this corank and, to the best of our knowledge, it is the first time that these examples can be found in physical contexts.

 We begin with the following notion, available in \cite{weibel}.

\begin{defn}\label{extensions}
A Lie algebra $\mathfrak{l}$ is an \textit{extension} of an abelian Lie algebra $\mathfrak{a}$ (called \textit{abelian kernel}) by another  Lie algebra $\mathfrak{b}$, if there is a short exact sequence
$$0 \ \longrightarrow \  \mathfrak{a} {\overset{\mu}{\longrightarrow}} \ \mathfrak{l} \ {\overset{\varepsilon}{\longrightarrow}} \mathfrak{b} \ \longrightarrow \ 0$$
such that  $\mathfrak{a}$ is an abelian ideal of $\mathfrak{l}$ and $\mathfrak{l}/\mathfrak{a} \simeq \mathfrak{b}$ is the Lie algebra quotient. Here $\mu$ is a monomorphism and $\varepsilon$ an epimorphism of Lie algebras such that $\ker \varepsilon = \mathrm{Im} \ \mu $. If in addition, $\mathfrak{a}$ is contained in the center \[Z(\mathfrak{l})=\{ a \in \mathfrak{l}  \ | \ [a,b]=0  \ \forall b \in \mathfrak{l}\} \]
 of $\mathfrak{l}$, the Lie algebra $\mathfrak{l}$ is called \textit{central extension} of $\mathfrak{a}$ by $\mathfrak{b}$.
\end{defn}

Definition \ref{extensions} is reported from \cite[Definition 7.6.1]{weibel}, where it is noted that $\mathfrak{a}$ may be endowed with the structure of $\mathfrak{b}$-module  defining $ba$ to be the product $[\widetilde{b},a]$ in $\mathfrak{l}$, where $\varepsilon(\widetilde{b})=b$ for all $a \in \mathfrak{a}$ and $b \in \mathfrak{b}$. Following \cite[Extension Problem 7.6.2]{weibel}, two extensions  $0 \ \longrightarrow \  \mathfrak{a} {\overset{\mu_1}{\longrightarrow}} \ \mathfrak{l}_1 \ {\overset{\varepsilon_1}{\longrightarrow}} \mathfrak{b} \ \longrightarrow \ 0$ and $0 \ \longrightarrow \  \mathfrak{a} {\overset{\mu_2}{\longrightarrow}} \ \mathfrak{l}_2 \ {\overset{\varepsilon_2}{\longrightarrow}} \mathfrak{b} \ \longrightarrow \ 0$ are equivalent if there is an isomorphism $\varphi$ such that the following diagram 
$$\begin{CD}
0  @>>>   \mathfrak{a} @>\mu_1>> \mathfrak{l}_1 @>\varepsilon_1>>  \mathfrak{b} @>>> 0\\ 
 @. @| @VV\varphi V @| @. \\ 
0  @>>>   \mathfrak{a} @>\mu_2>> \mathfrak{l}_2 @>\varepsilon_2>>  \mathfrak{b} @>>> 0
\end{CD}
$$
is commutative. Then we can define the set $\mathrm{Ext}(\mathfrak{b},\mathfrak{a})$ of equivalence classes of the extensions of  $\mathfrak{a} $ by  $\mathfrak{b}$ realizing $\mathfrak{l}$. The problem of the description of $\mathrm{Ext}(\mathfrak{b},\mathfrak{a})$ is known as Extension Problem and a  classical result \cite[Classification Theorem 7.6.3]{ weibel} shows an important relation with the Schur multiplier.

\begin{lem}[See \cite{weibel}, Chapter 7]\label{h2}Given an extension $0 \ \rightarrow \  \mathfrak{a} \rightarrow \ \mathfrak{l} \rightarrow \mathfrak{b} \ \rightarrow \ 0$ of a Lie algebra $\mathfrak{l}$ as in Definition \ref{extensions}:
\begin{itemize}
\item[(i)]There is a 1-1 correspondence between $M(\mathfrak{l})$ and $\mathrm{Ext}(\mathfrak{b},\mathfrak{a})$;
\item[(ii)] If $M(\mathfrak{l})=0$, then $\mathfrak{l}$ cannot be realized as  extension with abelian kernel.
\end{itemize}
\end{lem}

As indicated in Lemma \ref{h2} (ii),  the size of the Schur multiplier actually measures how many extensions are possible. Further important examples of extensions are constructed below.

\begin{defn}\label{semidirect}
A Lie algebra $\mathfrak{l}$ is the {\em semidirect sum} of two of its Lie subalgebras $\mathfrak{a}$ and $\mathfrak{b}$ if  the following conditions are satisfied:
\begin{itemize}
\item[(i)]$\mathfrak{a}$  is an ideal of $\mathfrak{l}$,
\item[(ii)] $\mathfrak{l}= \mathfrak{a} + \mathfrak{b}$,
\item[(iii)] $\mathfrak{a} \cap \mathfrak{b} = 0$.
\end{itemize}
Of course, if $\mathfrak{a} \subseteq Z(\mathfrak{l})$, then $\mathfrak{l}$ is a central extension of $\mathfrak{a}$ by $\mathfrak{b}$.
\end{defn}

Having in mind Lemma \ref{h2} (i),  Definition \ref{semidirect} and Definition \ref{extensions}, it is clear that semidirect sums are special cases of extensions with abelian kernel, and so the Schur multiplier gives information of structure of a Lie algebra. In this perspective, we recall another classical result. A semisimple Lie algebra is a sum of simple Lie algebras. A well known theorem of Levi illustrates the importance of the semidirect sum in the structure of the finite dimensional Lie algebras.

\begin{thm}[See \cite{weibel}, Chapter 7, Levi Decomposition]
Any finite dimensional Lie algebra $\mathfrak{l}$ over $\mathbb{C}$  is semidirect sum of a solvable Lie algebra $\mathfrak{a}$ by a semisimple Lie algebra $\mathfrak{b}$.
\end{thm}

 Note that a solvable Lie algebra is a Lie algebra which
 contains a finite series of ideals in which each factor is
 an abelian Lie algebra. While the role of the Schur
 multiplier of solvable Lie algebras is quite difficult to
 compute, the  result above is classical and may help for the
 case of semisimple Lie algebras. An important consequence is
 the following.

\begin{cor}[See \cite{weibel}, Chapter 7]\label{semisimple} A semisimple Lie algebra $\mathfrak{l}$ over $\mathbb{C}$ of finite dimension has $M(\mathfrak{l})=0$.
\end{cor}



For a series of motivations, which will be clear later on, we need to investigate Schur multipliers of Lie algebras in which we may iterate the notion of center.

\begin{defn}\label{center}If $\mathfrak{l}$ is a Lie algebra, the \textit{derived} Lie subalgebra of $\mathfrak{l}$ is defined as the smallest Lie ideal containing all the commutators  $[a,b]=ab-ba$  of $\mathfrak{l}$, that is,
\[[\mathfrak{l},\mathfrak{l}]=\langle [a,b]  \ | \ a,b \in \mathfrak{l} \rangle. \]
Note that the Lie algebra quotient $\mathfrak{l}/[\mathfrak{l},\mathfrak{l}]$ is always abelian,
 \end{defn}

Both  $Z(\mathfrak{l})$ and $[\mathfrak{l},\mathfrak{l}]$ are ideals of $\mathfrak{l}$ and so the size of their quotients may  measure  how far we are from having an abelian Lie algebra. Since we do not know whether $\mathfrak{l}/Z(\mathfrak{l})$ is abelian or not in general,  there is a way to generalize the above notions via appropriate series (see \cite{ weibel}).

\begin{defn}\label{uppercentralseries}The \textit{upper central series} of $\mathfrak{l}$ is the series
\[0=Z_0(\mathfrak{l}) \leq Z_1(\mathfrak{l})=Z(\mathfrak{l}) \leq Z_2(\mathfrak{l}) \leq \ldots \leq Z_i(\mathfrak{l}) \le Z_{i+1}(\mathfrak{l}) \leq \ldots, \]
where each $Z_i(\mathfrak{l})$ turns out to be an ideal of $\mathfrak{l}$ (called $i$th $center$ of $\mathfrak{l}$) and 
\[ \frac{Z_1(\mathfrak{l})}{Z_0(\mathfrak{l})}=Z(\mathfrak{l}), \frac{Z_2(\mathfrak{l})}{Z_1(\mathfrak{l})}=Z\left(\frac{\mathfrak{l}}{Z(\mathfrak{l})}\right),   \ldots, \frac{Z_{i+1}(\mathfrak{l})}{Z_i(\mathfrak{l})}=Z\left(\frac{\mathfrak{l}}{Z_i(\mathfrak{l})}\right), \ldots \]
is an abelian quotient Lie algebra. We say that $\mathfrak{l}$ is  \textit{nilpotent of class $c$} if $\mathfrak{l}=Z_c(\mathfrak{l})$, that is, if the upper central series of $\mathfrak{l}$ ends after a finite number $c \ge 0$ of steps.

By analogy, the \textit{lower central series} of $\mathfrak{l}$ is defined  inductively by
\[\gamma_1(\mathfrak{l})=\mathfrak{l} \geq \gamma_2(\mathfrak{l})=[\mathfrak{l},\mathfrak{l}]  \geq  \gamma_3(\mathfrak{l})=[[\mathfrak{l},\mathfrak{l}],\mathfrak{l}]  \geq \ldots \geq \gamma_i(\mathfrak{l}) \geq \gamma_{i+1}(\mathfrak{l}) \geq \ldots, \]
where each $\gamma_i(\mathfrak{l})$ turns out to be an ideal of $\mathfrak{l}$ (called $i$th $derived$ of $\mathfrak{l}$). 
\end{defn}

The presence of finite dimensional nilpotent Lie algebras plays a fundamental role in mathematical physics, because these structures have been observed in several quantum dynamical systems (see \cite{beltita}).  Probably one of the first nilpotent Lie algebras, which was used for the description of the behaviour of complex dynamical systems, in which bosons are involved, is due to Heisenberg and we report the formal definition here.

\begin{defn}\label{Heisenberg}  A finite dimensional Lie algebra  $\mathfrak{l}$ is called $Heisenberg$ provided that $[\mathfrak{l},\mathfrak{l}]=Z(\mathfrak{l})$ and $\mathrm{dim}([\mathfrak{l},\mathfrak{l}]) = 1$. Such algebras are odd dimensional with
basis $v_1, \ldots , v_{2m}, v$ and the only non--zero commutator between basis elements is $[v_{2i-1}, v_{2i}] = -
[v_{2i}, v_{2i-1}]= v$ for $i = 1,2, \ldots ,m$. The symbol $\mathfrak{h}(m)$
denotes the Heisenberg algebra of dimension $2m + 1$.
\end{defn}

 One can see easily that $\mathfrak{h}(m)$ is a nilpotent
 Lie algebra of class $2$, that is, such that
 $Z_2(\mathfrak{h}(m)))=\mathfrak{h}(m)$, of
 arbitrary dimension of the form $2m+1$,  and such that the
 Lie algebra quotient
 $\mathfrak{h}(m)/Z(\mathfrak{h}(m))$ is abelian,
 where $Z(\mathfrak{h}(m))$ a 1-dimensional Lie algebra.
 It turns out that the structure of finite dimensional
 nilpotent Lie algebras of low dimension involves strongly
 Heisenberg algebras and may be described in terms of
 generators and relations in the following way.

\begin{thm}[Classification of Finite Dimensional Nilpotent Lie Algebras of Dimension 3, 4 and  5]\label{classification}
Let $\mathfrak{l}$ be a finite dimensional nilpotent Lie algebra (over any field of characteristic $\neq 2$) and $\mathfrak{i}$ an abelian Lie algebra of dimension 1. Then
\begin{itemize}
\item[(1)] $\mathrm{dim} \ \mathfrak{l} =3$ if and only if $\mathfrak{l}$ is isomorphic to one of the following Lie algebras:
 \begin{itemize}
\item[-] $\mathfrak{l}_{3,1}$ abelian Lie algebra of dimension 3,
\item[-] $\mathfrak{l}_{3,2}  \simeq \mathfrak{h}(1)$.
\end{itemize}
\item[(2)]$\mathrm{dim} \ \mathfrak{l} =4$ if and only if    $\mathfrak{l}$ is isomorphic to one of the following Lie algebras:
 \begin{itemize}
\item[-] $\mathfrak{l}_{4,1} =  \mathfrak{l}_{3,1} \oplus \mathfrak{i}$,
\item[-] $\mathfrak{l}_{4,2} =  \mathfrak{l}_{3,2} \oplus \mathfrak{i}$
\item[-] $\mathfrak{l}_{4,3} = \langle v_1, v_2, v_3, v_4 \ | \ [v_1,v_2]=v_3, [v_1,v_3]=v_4 \rangle$
\end{itemize}
\item[(3)]$\mathrm{dim} \ \mathfrak{l} =5$ if and only if $\mathfrak{l}$ is isomorphic to one of the following Lie algebras:
\begin{itemize}
\item[-] $\mathfrak{l}_{5,1} = \mathfrak{l}_{4,1} \oplus \mathfrak{i} $,
\item[-] $\mathfrak{l}_{5,2} = \mathfrak{l}_{4,2} \oplus \mathfrak{i} \simeq \mathfrak{h}(1) \oplus \mathfrak{i} \oplus \mathfrak{i}$,
\item[-] $\mathfrak{l}_{5,3} = \mathfrak{l}_{4,3} \oplus \mathfrak{i} $,
\item[-] $\mathfrak{l}_{5,4} =  \langle v_1, v_2, v_3, v_4, v_5 \ | \ [v_1,v_2]=[v_3,v_4]=v_5 \rangle \simeq \mathfrak{h}(2) $,
\item[-] $\mathfrak{l}_{5,5} =  \langle v_1, v_2, v_3, v_4, v_5 \ | \ [v_1,v_2]=v_3, [v_1,v_3]=[v_2,v_4]=v_5 \rangle$,
\item[-] $\mathfrak{l}_{5,6} =  \langle v_1, v_2, v_3, v_4, v_5 \ | \ [v_1,v_2]=v_3, [v_1,v_3]=v_4,$ $$[v_1,v_4]=[v_2,,v_3]=v_5  \rangle,$$
\item[-] $\mathfrak{l}_{5,7} =  \langle v_1, v_2, v_3, v_4, v_5 \ | \ [v_1,v_2]=v_3, [v_1,v_3]=v_4, [v_1,v_4]=v_5 \rangle$,
\item[-] $\mathfrak{l}_{5,8} =  \langle v_1, v_2, v_3, v_4, v_5 \ | \ [v_1,v_2]=v_4, [v_1,v_3]=v_5 \rangle$,
\item[-] $\mathfrak{l}_{5,9} =  \langle v_1, v_2, v_3, v_4, v_5 \ | \ [v_1,v_2]=v_3, [v_1,v_3]=v_4, [v_2,v_3]=v_5 \rangle$.
\end{itemize}
\end{itemize}
\end{thm}

An important observation can be done here: this is designed for our purposes.

\begin{lem}\label{technical}
The nilpotent Lie algebras $\mathfrak{l}_{4,3}$, $\mathfrak{l}_{5,4}$, $\mathfrak{l}_{5,5}$,  $\mathfrak{l}_{5,6}$, $\mathfrak{l}_{5,7}$, $\mathfrak{l}_{5,8}$ and $\mathfrak{l}_{5,9}$ may be realized as central extensions.
\end{lem}

\begin{proof} We begin to consider $\mathfrak{l}_{4,3}$, which is  described at \cite[Page 3530]{har1}. We note that  $\mathfrak{l}_{4,3}$ is nilpotent of class $2$ and is realized as central extension of a  1-dimensional ideal $\mathfrak{i}=\mathfrak{a} \subseteq Z(\mathfrak{l}_{4,3})$ by the  quotient  $\mathfrak{b}=\mathfrak{l}_{4,3}/\mathfrak{a}=\mathfrak{h}(1)$.  Definition \ref{extensions} is satisfied. On the other hand, it is clear that the Heisenberg algebra of dimension 5 is realized as central extension, just looking at the way in which it is defined. Then we look at \cite[Pages 3531--3532, Case 2]{har1} and see that $\mathfrak{l}_{5,5}$ is also nilpotent of class $2$ and  realized as central extension of a  1-dimensional ideal $\mathfrak{i}=\mathfrak{a} \subseteq Z(\mathfrak{l}_{5,5})$ by the Lie algebra quotient  $\mathfrak{b}=\mathfrak{l}_{5,5}/\mathfrak{a}=\mathfrak{h}(1)+\mathfrak{i}$. Again Definition \ref{extensions} is satisfied. On the other hand, we may look at \cite[Pages 3531--3532]{har1} and see that $\mathfrak{l}_{5,8}$ may be  realized as central extension of a  1-dimensional ideal $\mathfrak{i}=\mathfrak{a} \subseteq Z(\mathfrak{l}_{5,5})$ by the Lie algebra quotient  $\mathfrak{b}=\mathfrak{l}_{5,5}/\mathfrak{a}=\mathfrak{h}(1)+\mathfrak{i}$. Again Definition \ref{extensions} is satisfied.  Then we consider $\mathfrak{l}_{5,6}$. We  look  at \cite[Page 4206, Case 2.1]{har} and see that it is nilpotent of class $2$ and realized as central extension of a 1-dimensional ideal  $\mathfrak{i}=\mathfrak{a} \subseteq \mathfrak{l}_{5,6}=\mathfrak{l}_{4,3}$. Definition \ref{extensions} is satisfied and, in the same way, we can construct  $\mathfrak{l}_{5,7}$ and $\mathfrak{l}_{5,9}$ as central extensions with central ideal of dimension one. These three central extensions of a central ideal of dimension one by $\mathfrak{l}_{4,3}$ are not equivalent, as indicated by the computations at \cite[Page 4206, Case 2.1]{har}. The result follows.
\end{proof}

 We will quote now some recent results on Schur multipliers of Lie algebras, in order to understand numerically the size of $M(\mathfrak{l})$

\begin{thm}[See \cite{nr1}, Lemmas 2.4 and 2.5] \label{ab} \

\begin{itemize}
\item[(i)]$\mathrm{dim} \ M(\mathfrak{h}(1))=2$.
\item[(ii)]$\mathrm{dim} \ M(\mathfrak{h}(m)) =2m^2-m-1$ for all $m\geq 2$.
\item[(iii)]A nilpotent finite dimensional Lie algebra $\mathfrak{l}$ of dimension $n$ has
$$\mathrm{dim} \ M(\mathfrak{l}) \le \frac{1}{2}n(n-1)$$
and the bound is exactly $\frac{1}{2}n(n-1)
$ if and only if $\mathfrak{l}$ is abelian.
\end{itemize}
\end{thm}

The condition (iii) above is due to Moneyhun \cite{mon} and allows us to introduce the following notion, which has been studied in \cite{nr3}.

\begin{defn}\label{corank} Given a nilpotent Lie algebra $\mathfrak{l}$ of dimension $n$, there exists a (unique) integer  $t(\mathfrak{l})\geq 0$ such that
\begin{equation}
t(\mathfrak{l})=\frac{1}{2}n(n - 1) - \mathrm{dim}\ M(\mathfrak{l}).
\end{equation}
This number is called \textit{corank} of $\mathfrak{l}$ (in the sense of Berkovich--Zhou).
\end{defn}
The  invariant of Definition \ref{corank} is originally due to Berkovich \cite{berkovich} and Zhou \cite{zhou}, who  introduced it for finite $p$-groups. It clearly gives information on the gap between the size of $M(\mathfrak{l})$ and the critical value $\frac{1}{2}n(n - 1)$ in which we find abelian Lie algebras only (see Theorem \ref{ab}). The same information can be rephrased in the context of groups as well. Of course, the case of Lie algebra is more rich than the case of finite $p$-groups and we have more instruments than the group case; this has allowed to describe the structure of all nilpotent Lie algebras in \cite{har} and \cite{har1} with $t(\mathfrak{l}) = 3, 4, 5, 6, 7 ,8$.

 \begin{thm}[Classification of Finite Dimensional Nilpotent Lie algebras via $t(\mathfrak{l}) \le 6$]\label{tl} Let $\mathfrak{l}$ be an $n$-dimensional nilpotent Lie algebra. Then
\begin{itemize}
\item[(i)]$t(\mathfrak{l})=0$ if and only if $\mathfrak{l}$ is abelian;
\item[(ii)]$t(\mathfrak{l})=1$ if and only if $ \mathfrak{l} \simeq \mathfrak{h}(1)$;
\item[(iii)]$t(\mathfrak{l})=2$ if and only if $ \mathfrak{l}  \simeq \mathfrak{l}_{4,2} \simeq \mathfrak{h}(1) \oplus \mathfrak{i}$;
\item[(iv)]$t(\mathfrak{l})=3$ if and only if $ \mathfrak{l} \simeq \mathfrak{l}_{5,2} \simeq \mathfrak{l}_{4,2} \oplus \mathfrak{i}$;
\item[(v)]$t(\mathfrak{l})=4$ if and only if either $ \mathfrak{l} \simeq \mathfrak{l}_{5,2} \oplus \mathfrak{i}$, or $ \mathfrak{l} \simeq  \mathfrak{l}_{4,3}$ or $ \mathfrak{l} \simeq \mathfrak{l}_{5,8}$;
\item[(vi)] $t(\mathfrak{l})=5$ if and only if $ \mathfrak{l} \simeq \mathfrak{l}_{7,2} \simeq \mathfrak{l}_{5,2}\oplus \mathfrak{i} \oplus \mathfrak{i}$ or $ \mathfrak{l} \simeq \mathfrak{h}(2)$.
\item[(vii)] $t(\mathfrak{l})=6$ if and only if either  $ \mathfrak{l} \simeq \mathfrak{l}_{4,2}\oplus \mathfrak{i}$, or $\mathfrak{l} \simeq \mathfrak{l}_{5,5}$, or $ \mathfrak{l} \simeq \mathfrak{h}(2) \oplus \mathfrak{i}$, or $\mathfrak{l} \simeq \mathfrak{l}_{5,8} \oplus \mathfrak{i}$, or $\mathfrak{l} \simeq \mathfrak{l}_{8,2} \simeq \mathfrak{l}_{7,2} \oplus \mathfrak{i}$;
\end{itemize}
\end{thm}




There are a series of open problems which prevent us to have a more general approach for higher dimensions. 

\begin{op}The classification for $t ( \mathfrak{l} ) \ge 9$ is not known in general. There are some partial classification, imposing restrictions of various nature on $\mathfrak{l}$. 
\end{op}

On the other hand, \cite{har} shows that for  $t ( \mathfrak{l} ) \in \{7,8\}$ we have nilpotent Lie algebras of dimension higher than $5$ for which an argument as in  Lemma \ref{technical} is not available. Therefore the main results of this paper will show that it is possible to realize physically the above nilpotent Lie algebras using concepts of quantum physics. This is the first time that the classification above may be interpreted in terms of concrete dynamical systems, and we will see that this approach is so sophisticated that may cover in fact all the Lie algebras of Theorem \ref{tl}.

\section{Main results}
In the present section we mainly refer to pseudo-bosons, considering bosons as a particular case of those and assume the reader is familiar with the terminology which is used in \cite{baginbagbook, a2, bagBS, bagnew1, bagnew2, barnett, gold, schu}. 

In particular, one needs to have clear the notions of  

\begin{itemize}
\item[-] bounded operators in Hilbert spaces (see \cite[Definition 3.1]{bagrus2018})  
\item[-]Riesz basis (see \cite{schu}), 
\item[-] $\mathcal{G}$-quasi basis in \cite[Definition 3.3]{bagrus2018}, 
\item[-]canonical commuting relations (see \cite{baginbagbook}), 
\item[-] distribution (see \cite[Eq. (3.2)]{bagrus2018}), 
\item[-]Assumption $\mathcal{D}$-pb 1,  Assumption $\mathcal{D}$-pb 2, Assumption $\mathcal{D}$-pb 3, Assumption $\mathcal{D}$-pbw 3, Assumption $\mathcal{D}$-pbs 3, reported in \cite[Section 3]{bagrus2018}. 
\end{itemize}

The above concepts are strictly related to the field of the functional analysis and allow us to construct the framework  in which pseudo-bosons may form the structure of a Lie algebra, which is described in the previous section in abstract.
 

\begin{thm}\label{main1}There is at least one family of  pseudo-bosonic operators and a nilpotent Lie algebra $\mathfrak{l}$ of finite dimension for arbitrary values of $t(\mathfrak{l})$.
\end{thm}

\begin{proof} About the notations and the terminology for pseudo-bosonic operators, we will use the same in \cite[Section 3]{bagrus2018}. We begin to consider the nilpotent Lie algebra \begin{equation}
\mathfrak{l} = \mathfrak{h}(m) \oplus \underbrace{\mathfrak{i} \oplus \ldots \oplus \mathfrak{i}}_{k-\mbox{times}}.
\end{equation}
This Lie algebra has dimension $(2m+1)+k$, where $m \ge 1$ and $k \ge 0$ are arbitrary integers.

If $k=0$, then $\mathfrak{l} \simeq \mathfrak{h}(m)$. In this case a possible example  can be constructed with $m$-dimensional bosonic operators, and with their pseudo-bosonic versions.

Let $x_j$, $j=1,2,\ldots,m$ be $m$ position operators, and let $p_j=-i\frac{d}{dx_j}$ the related momentum operator. Then the generators of $\mathfrak{h}(m)$ can be written as follows:
$$
v_{2j-1}=\frac{1}{\sqrt{2}}(x_j+ip_j), \qquad v_{2j}=\frac{1}{\sqrt{2}}(x_j-ip_j),
$$
$j=1,2,\ldots,m$, with $v_{2m+1}=\mathbb{I}$, the identity operator. There are just lowering and raising operators associated to an $m$-dimensional quantum harmonic oscillator, or (alternatively) to $m$ modes of bosons.

A different representation can also be constructed easily in terms of pseudo-bosonic operators, simply by considering a shifted version of the $v_j$, that is,

$$w_{2j-1}=v_{2j-1}+\alpha_j \mathbb{I}, \qquad  w_{2j}=v_{2j}+\beta_j \mathbb{I},$$  
$j=1,2,\ldots,m,$, for any scalar $\alpha_j, \beta_j \in \mathbb{C}$ with $\alpha_j\neq\overline{\beta_j}.$ 

On the other hand,  we may consider the genuine abelian case $\mathfrak{l}$ of dimension $k$ putting $m=0$. Therefore
\begin{equation}\label{heisenberg+abelian}
\mathfrak{l}=\langle v_1, v_2, \ldots, v_{2m}, v_{2m+1}, \ldots, v_{2m+k+1} \ | \ [v_{2j-1},v_{2j}]=\mathbb{I} \  \ \forall j =1, 2, \ldots, m,
\end{equation}
$$ \mbox{and} \ v_{2m+1}= v_{2m+2}= \ldots = v_{2m+k+1}=\mathbb{I} \rangle. $$

This means that there is no loss of generality in assuming both $m \ge 0$ and $k\ge 0$.

From \cite[Theorem 4.2, Page 44]{nr3}, we can see that
\begin{equation}
t(\mathfrak{l}) = \mathrm{dim} \ \mathfrak{l} = 2m+k+1,
\end{equation}
where $m$ and $k$ can grow arbitrarily.
This allows us to conclude that $\mathfrak{l}$ is always realized by pseudo-bosons and its corank may be arbitrary big, so the result follows.

\end{proof}


Now we provide concrete examples for the Lie algebras with $t(\mathfrak{l}) \le 6$.


\begin{thm}\label{main2}All the nilpotent Lie algebras of Theorem \ref{tl} are realized by pseudo-bosonic operators.
\end{thm}

\begin{proof}
We begin to consider $t(\mathfrak{l})$ for the values $0,1,2,3,4$. Looking at Theorem \ref{tl} (i), (ii), (iii), (iv) and (v), we identify just two cases which do not fit into the form of the Lie algebra \eqref{heisenberg+abelian}. These are $\mathfrak{l}_{4,3}$ and $\mathfrak{l}_{5,8}$. We are going to prove that both of them may be described in terms of pseudo-bosonic operators. About the notations and the terminology for pseudo-bosonic operators, we will use again that  in \cite[Section 3]{bagrus2018}.

Let $c$ be a bosonic operator such that $[c,c^\dagger]=cc^\dagger-c^\dagger\,c=\mathbb{I}$.  We define
\begin{equation}\label{add1}
v_1=c,\quad v_2=\frac{1}{2}\,(c^\dagger)^2,\quad v_3=c^\dagger,\quad v_4=\mathbb{I}.
\end{equation}
It is easy to check that all the commutators are zero, except that
$$ [v_1,v_2]=v_3, \qquad [v_1,v_3]=v_4.$$
Hence we get $\mathfrak{l}_{4,3}$. These operators can be used in quantum optics to construct two-photons Hamiltonians which are relevant, for instance, in connection with squeezed states, \cite{barnett}. The Hamiltonian is $H_0=\omega c^\dagger c+i(g{c^\dagger}^2-\overline{g}c^2)$ which, in terms of $v_j$ above can be written as
$$
H=\omega v_3v_1+i\left(2gv_2-\overline{g}v_1^2\right).
$$
Here $\omega$ is a real quantity related to the free energy of the bosons, while the complex $g$ {\em measures} the degree of squeezing of the system, \cite{barnett}.

Another example of the same algebra can be constructed using pseudo-bosonic operators $a$ and $b$ satisfying $[a,b]=\mathbb{I}$. In this case we put
$$
v_1=a,\quad v_2=\frac{1}{2}\,b^2,\quad v_3=b,\quad v_4=\mathbb{I}.
$$ They also describe $\mathfrak{l}_{4,3}$.

Of course, the same Hamiltonian $H$ as above could be used simply replacing $v_j$ with their pseudo-bosonic counterparts. This could be relevant in the analysis of bi-squeezed states, which extend squeezed states to a non hermitian context, as bi-coherent states do for coherent states.

Now we pass to describe $\mathfrak{l}_{5,8}$.
In this case one single bosonic or pseudo-bosonic operator is not enough. For this reason, let us consider two bosonic operators $c_1$ and $c_2$ such that $$
[c_k,c_l^\dagger]=\delta_{k,l}\mathbb{I}, \qquad [c_k,c_l]=0,
$$
where $\delta_{k,l}$ denotes the well known symbol of Kronecker. Then, defining
$$
v_1=c_1,\quad v_2=c_1^\dagger,\quad v_3=c_1^\dagger\,c_2,\quad v_4=\mathbb{I}, \quad v_5=c_2,
$$
the only non zero commutators are the following:
$$
[v_1,v_2]=v_4, \qquad [v_1,v_3]=v_5.$$
Hence we recover $\mathfrak{l}_{5,8}$. The same result can be found if we use two pseudo-bosonic operators $[a_1,b_1]=[a_2,b_2]=\mathbb{I}$, all the other commutators being zero. In this case we have to define
$$ v_1=a_1,\quad v_2=b_1,\quad v_3=b_1\,a_2,\quad v_4=\mathbb{I}, \quad v_5=a_2.
$$
As for a physical appareance of these (bosonic or pseudo-bosonic) operators we could consider the operator
 $$
 H=\lambda(v_3+v_2v_5),
 $$
for real $\lambda$. $H$ is manifestly not self-adjoint either if we use bosonic or pseudo-bosonic operators. The meaning of $H$ is simple: it describes interactions between the two modes of particles. Notice that no kinetic term is present in $H$. This is not strange, in concrete models, since quite often the most relevant part of an Hamiltonian is that term which describes the interactions. 

Then we pass to consider the case of $t(\mathfrak{l})=5$, looking at Theorem \ref{tl} (vi). Again here we may realize pseudo-bosonic operators via the Lie algebra \eqref{heisenberg+abelian}.

Passing to the case $t(\mathfrak{l})=6$,  Theorem \ref{tl} (vii) shows that $\mathfrak{l}_{5,5}$ and $\mathfrak{l}_{5,8} \oplus \mathfrak{i}$ are the cases not involved in \eqref{heisenberg+abelian}. While $\mathfrak{l}_{5,8} \oplus \mathfrak{i}$ may be written easily in terms of pseudo-bosonic operators adding an operator which commutes with all those of $\mathfrak{l}_{5,8}$, an explicit construction for $\mathfrak{l}_{5,5}$  must be given.

For that, we can invoke Lemma \ref{technical} and see that it is possible to construct this Lie algebra with a family of pseudo-bosonic operators realizing  the Lie algebra $\mathfrak{h}(1) \oplus \mathfrak{i}$ and a central ideal $\mathfrak{a}=\mathfrak{i}$. Therefore $\mathfrak{l}_{5,5}$ is the central extension of $\mathfrak{a}$ by $\mathfrak{b}=\mathfrak{l}_{5,5}/\mathfrak{a}$. This is a first possible abstract approach.

A different realization of $\mathfrak{l}_{5,5}$ can be constructed using again two modes of bosonic (or pseudo-bosonic) operators $c_1$ and $c_2$ as for $\mathfrak{l}_{5,8}$. Now we define
$$
v_1=c_1,\quad v_2=c_2+\frac{1}{2}(c_1^\dagger)^2,\quad v_3=c_1^\dagger,\quad v_4=c_2^\dagger,\quad v_5=\mathbb{I}.
$$
The only non zero commutators are $[v_1,v_2]=v_3$, and $[v_1,v_3]=[v_2,v_4]=v_5$. These commutators can also be recovered using pseudo-bosonic operators $(a_1,b_1)$ and $(a_2,b_2)$. In this case we define
$$
v_1=a_1,\quad v_2=a_2+\frac{1}{2}b_1^2,\quad v_3=b_1,\quad v_4=b_2,\quad v_5=\mathbb{I}.
$$
A possible Hamiltonian constructed out of these last operators is the following:
$$
H=\omega v_3v_1+\lambda v_4v_2,
$$
which is manifestly non self-adjoint, and describes a sort of "always increasing"
 interaction between two modes, without any decreasing.
\end{proof}

As we have seen in the proof of the previous theorem, the notion of central extension in Definition \ref{extensions} is appropriate to describe nilpotent Lie algebras in terms of pseudo-bosons. In fact we can describe all those of dimension at most five.

\begin{cor}\label{dimension}
All the nilpotent Lie algebras of Theorem \ref{classification} may be realized by central extensions of families of pseudo-bosonic operators.
\end{cor}

\begin{proof}
Using Lemma \ref{technical} and the fact that \eqref{heisenberg+abelian} is described by pseudo-bosons, the result follows.
\end{proof}

 The conclusions of the above corollary may be rephrased in a
 more significant way, showing the the theory of the
 pseudo-boson operators is really a powerful tool both in
 mathematics and physics.

 \begin{cor}\label{final}
 If there exists a finite dimensional nilpotent Lie algebras
 which cannot be realized by central extensions of families
 of pseudo-bosonic operators, then its dimension must be
 $\ge 6$.
 \end{cor}
 
\section*{Acknowledgements}
F.B. acknowledges partial support from Palermo University and from the Gruppo Nazionale di Fisica Matematica (GNFM) of the Istituto Nazionale di Alta Matematica (INdAM).

\end{document}